\documentclass[prl,aps,reprint,amsmath,amssymb]{revtex4-1}

\usepackage[utf8]{inputenc}
\usepackage[T1]{fontenc}
\usepackage[english]{babel}

\usepackage{graphicx}
\usepackage{bm}
\usepackage[percent]{overpic}
\usepackage{mathtools}
\usepackage[unicode]{hyperref}
\usepackage{epsfig,amsthm,mathrsfs,dsfont,ctable,url,commath}
\usepackage{color}
\usepackage{braket}
\usepackage{units}
\usepackage{amscd}
\usepackage{thmtools}

\usepackage[inline]{enumitem}
\usepackage[smaller,nolist]{acronym}

\newtheorem{theorem}{Theorem}
\theoremstyle{remark}
\newtheorem*{remark}{Remark}

\newcommand{\orcid}[1]{\href{#1}{\includegraphics[height=.8em]{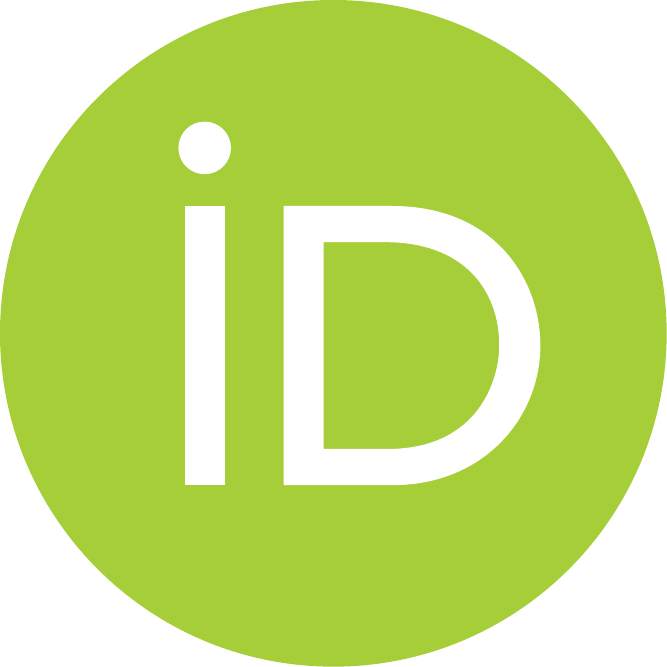}}}

\def\Norm#1{\|{#1}\|}

\def\>{\rangle}
\def\<{\langle} 

   \def\rA{{\rm A}}\def\rB{{\rm
    B}}

\def\Stset{{\mathsf{St}}}

 \def\Span{\mathsf{Span}}
  
  \def\Tr{{\rm Tr}}

\usepackage{array}
\newcolumntype{?}{!{\vrule width 1.2pt}}

\usepackage{soul}

\definecolor{blue1}{rgb}{0.03, 0.27, 0.49}

\usepackage{pifont}
\usepackage{xcolor}

\DeclareMathOperator{\Supp}{Supp}{}

\def\FockSpace{\mathcal{F}}
\def\Perr{\mathcal{P}_{\text{err}}}
\def\LOCCSet{\mathsf{LOCC}}
\def\SEPSet{\mathsf{SEP}}

\begin{document}

\title{Fermionic state discrimination by local operations and classical communication}

\author{Matteo \surname{Lugli} \orcid{http://orcid.org/0000-0003-0554-3760}}

\email{matteo.lugli01@ateneopv.it}

\affiliation{QUIT group, Dipartimento di Fisica, Università di Pavia, and INFN Sezione di Pavia, via Bassi 6, 27100 Pavia, Italy}

\author{Paolo \surname{Perinotti} \orcid{http://orcid.org/0000-0003-4825-4264}}

\email{paolo.perinotti@unipv.it}

\affiliation{QUIT group, Dipartimento di Fisica, Università di Pavia, and INFN Sezione di Pavia, via Bassi 6, 27100 Pavia, Italy}

\author{Alessandro \surname{Tosini} \orcid{http://orcid.org/0000-0001-8599-4427}}

\email{alessandro.tosini@unipv.it}

\affiliation{QUIT group, Dipartimento di Fisica, Università di Pavia, and INFN Sezione di Pavia, via Bassi 6, 27100 Pavia, Italy}

\begin{abstract}
    \acrodef{LOCC}{local operations and classical communication}
    \acrodef{POVM}{positive-operator valued measure}
    \acrodef{SEP}{separable effects}\acused{SEP}
    We consider the problem of \ac{LOCC} discrimination between two bipartite pure states of Fermionic systems.
    We show that, contrarily to the case of quantum systems, for Fermionic systems it is generally not possible to achieve the ideal state discrimination performances through \ac{LOCC} measurements.
    On the other hand, we show that an ancillary system made of two Fermionic modes in a maximally entangled state is a sufficient additional resource to attain the ideal performances via \ac{LOCC} measurements.
    The stability of the ideal results is studied when the probability of preparation of the two states is perturbed, and a tight bound on the discrimination error is derived.
\end{abstract}

\maketitle

The very concept of quantum information theory requires encoding distinguishable pieces of information on quantum states.
In the simplest instance of encoding of classical information, the decoding procedure corresponds to the widely studied task of quantum state discrimination~\cite{Helstrom1969,Helstrom1976,Ivanovic1987,Dieks1988,Peres1988,Chefles2000,Bergou2007,Barnett2009}.
In turn, the state discrimination task has been extensively studied in the scenario where states are shared by distant agents that are only allowed to use \acf{LOCC}~\cite{Walgate2000,Virmani2001,Jonathan1999}.
These tasks are now exhaustively understood in the quantum realm.

On the other hand, real physical systems are Bosons or Fermions, and the latter are ruled by a theory that is a slight variation of the standard quantum one.
The study of information processing in Fermionic theory has then various reasons, that are both practical and fundamental~\cite{Bravyi2002}.
Of particular importance is establishing analogies and differences between quantum and Fermionic implementation of specific information processing tasks.
For example, it is known that quantum and Fermionic computation are equivalent, meaning that every quantum algorithm can be efficiently mapped to a Fermionic one, and viceversa~\cite{Bravyi2002}.
This implies, e.g., that Fermionic processes are efficiently simulated by quantum computers~\cite{Jordan2014}.
In many other respects, however, the two theories present significant differences~\cite{Wolf2006,Banuls2007,DAriano2014a}.

In the present Letter, we study the task of \ac{LOCC} state discrimination in the Fermionic theory.
We show that, unlike the quantum case, in the typical situation \ac{LOCC} discrimination is strictly suboptimal.
We also derive conditions where ideal discrimination performances can be achieved via a \ac{LOCC} protocol.
These conditions are very sensitive to prior information about the probability of occurrence of the two states.
Therefore, we study the behavior of \ac{LOCC} protocols in the presence of a small perturbation of the ideal conditions.
Moreover, we show that a pair of Fermionic systems in a maximally entangled state is a sufficient resource in addition to \ac{LOCC} to achieve discrimination performances equivalent to the optimal one.

We briefly introduce the Fermionic quantum theory as the theory dealing with systems made of local Fermionic modes~\cite{Bravyi2002,Wolf2006,Banuls2007,Friis2013,DAriano2014}.
A Fermionic mode represents the counterpart of a qubit in the quantum theory and can be either empty or occupied by a single ``excitation.''
The states of Fermionic systems satisfy the \emph{parity superselection rule}~\cite{Wick1952,Hegerfeldt1968,Schuch2004,Kitaev2004,Schuch2004a,DAriano2014,DAriano2014a,Bravyi2002}, i.e., superpositions of vectors having even or odd excitation numbers are forbidden.
The latter can be derived as a consequence of the assumption that the elements of the Fermionic algebra are
Kraus operators of local Fermionic transformations~\cite{DAriano2014}.
The generators of the Fermionic algebra $\varphi_i$, $i$ running over arbitrary sets of  $N$ modes, fulfil the canonical anticommutation relations $\{\varphi_i, \varphi_j^\dagger\} = \delta_{ij}$ and $\{\varphi_i, \varphi_j\} = \{\varphi_i^\dagger, \varphi_j^\dagger\} = 0 \ \forall i, j$.
Once we define the vacuum state $\ket{\Omega}$ as the common eigenvector of operators $\varphi^\dag_i\varphi_i$ with null eigenvalues, the Fermionic operators enable us to define the Fock states as $\ket{n_1\ldots n_N} \coloneqq (\varphi_1^\dagger)^{n_1} \cdots (\varphi_N^\dagger)^{n_N} \ket{\Omega}$ and the antisymmetrized Fock space $\FockSpace$ through the linear combination of all Fock states.
We may label with the lowercase letters $e$, $o$ those sectors of the Fock space featuring even and odd parity, respectively.
The Jordan-Wigner isomorphism~\cite{Jordan1928,Verstraete2005,Pineda2010} is a crucial tool to handle the transformations and informational protocols in Fermionic theory.
Indeed, it maps non-locally the Fermionic operator algebra to an algebra of transformations on qubits, thus allowing us to proceed with the usual quantum notation.

\section{The orthogonal case}
In quantum theory, we may perfectly discriminate between any two orthogonal states $\ket\psi$, $\ket\phi$ of a bipartite system $\rA\rB$ through \ac{LOCC} measurements~\cite{Walgate2000}.
We remind that the most general case of a quantum measurement is represented by a \ac{POVM}, i.e., a collection of effects (positive operators $0 \le S \le I$) that sum to the identity operator $I$.
A necessary condition for a \ac{POVM} to represent a \ac{LOCC} measurement is to be separable (\ac{SEP}).
The effect $S$ is separable if there exists some operators $0 \le A_i, B_i \le I$ such that $S = \sum_i A_i \otimes B_i$, and a \textsc{povm} represents a separable measurement if it is exclusively made of separable effects.
Moreover, we recall that \ac{LOCC} \acp{POVM} are a proper subset of \ac{SEP} \acp{POVM}~\cite{Bennett1999}.
In the following we will use the acronyms \ac{LOCC} and \ac{SEP} to denote the corresponding subsets of \acp{POVM}.

We now give a sketchy summary of the result of Ref.~\cite{Walgate2000}. 
Every pair of orthogonal bipartite pure states can then be written as
\begin{equation}\label{eq:walg}
	\ket{\psi} =\sum_{i = 1}^n \ket{i}_\rA \ket{\eta_i}_\rB,\quad \ket{\phi} =\sum_{i = 1}^n \ket{i}_\rA \ket{\nu_i}_\rB ,
\end{equation}
where $\{\ket{i}_\rA\}$ is a suitable orthonormal basis in the Hilbert space of Alice's system, and 
$\{\ket{\eta_i}_\rB\}$ and $\{\ket{\nu_i}_\rB\}$ are sets of vectors in Bob's Hilbert space that are pairwise orthogonal, i.e., $\braket{\eta_i|\nu_i}=0$.
Alice has to measure her system in the given basis and send the outcome to Bob, who in turn manages to locally discriminate between two orthogonal states, thus inferring the correct result.
Existence of the decomposition in Eq.~\eqref{eq:walg} was shown in Ref.~\cite{Walgate2000}.

We follow here a strategy similar to the quantum one in order to distinguish between two pure orthogonal states $\ket\psi$, $\ket\phi$ of a bipartite Fermionic system.
First of all, we notice that whenever the two preparations have different parity, e.g., $\ket{\psi} \in \FockSpace_e (\rA\rB)$ and $\ket{\phi} \in \FockSpace_o (\rA\rB)$, it is always possible to perfectly discriminate between the two just through local measurements.
Indeed, Alice and Bob have to locally measure the parity of their subsystems and if their outcomes match, then the provided state was even, otherwise it was the odd one.
The nontrivial case then is that of two pure states with the same parity.
Since the even and odd sector are equivalent under \ac{LOCC}, it is not restrictive to focus on even vectors only.
We introduce the following convenient notation for the even vectors $\ket{\psi}, \ket{\phi} \in \FockSpace_e(\rA\rB)$,
\begin{equation}
    \begin{split}
    	&\ket{\psi} = \ket{\psi_E}+\ket{\psi_O}, \\
    	& \ket{\phi} = \ket{\phi_E} + \ket{\phi_O},
    \end{split}\label{eq:fermi-walg}
\end{equation}
and recalling the decomposition in Ref.~\cite{Walgate2000}, we decompose
$\ket{\psi_E}=\sum_{i = 1}^n \ket{e_i}_\rA \ket{\eta^e_i}_\rB $, $\ket{\psi_O}=\sum_{i = 1}^n \ket{o_i}_\rA \ket{\eta^o_i}_\rB $, $\ket{\phi_E}=\sum_{i = 1}^n \ket{e_i}_\rA \ket{\nu^e_i}_\rB $, $\ket{\phi_O}=\sum_{i = 1}^n \ket{o_i}_\rA \ket{\nu^o_i}_\rB $, where $\{\ket{e_i}\}$, $\{\ket{o_i}\}$ are Alice orthonormal bases of even and odd vectors, respectively, while $\{\ket{\eta^x_i}_\rB\}$, and $\{\ket{\nu^x_i}_\rB\}$, for $x=e,o$ are Bob vectors resulting from the decomposition.
In general, the latter are not normalized and $\braket{\eta_i^x|\nu_i^x}_\rB\neq0$.
We may indicate with the capitalized letters $E$ or $O$ those entities pertaining to the $E$ and $O$ spaces of $\FockSpace_e (\rA\rB)$, i.e., those subspaces where the parities of Alice's and Bob's subsystems are both even or odd, respectively.
E.g., the $E$ part of vector $\ket{\psi}$ is defined as $\ket{\psi_E} = \sum_{i = 1}^n \ket{e_i}_\rA \ket{\eta^e_i}_\rB$, whereas the $O$ part is $\ket{\psi_O} = \sum_{i = 1}^n \ket{o_i}_\rA \ket{\nu^o_i}_\rB$.
The orthogonality condition $\braket{\psi|\phi}=0$ generally reads
\begin{equation}\label{eq:orthogen}
    \braket{\psi_E|\phi_E}+\braket{\psi_O|\phi_O}=0.
\end{equation}

Let us consider as the first case the scenario where the two preparations have only one component. Then  they have components either in complementary subspaces, e.g., $\ket{\psi} = \ket{\psi_E}$ and $\ket{\phi} = \ket{\phi_O}$, and it is trivially possible to discriminate via \ac{LOCC} by measuring the local parities, or in the same subspace, e.g., $\ket{\psi} = \ket{\psi_E}$ and $\ket{\phi} = \ket{\phi_E}$.
In the latter case the protocol reduces to the quantum one.
Indeed, Alice selects the right basis $\{\ket{e_i}\}$ and lets Bob perfectly discriminate between $\ket{\eta^e_i}$ and $\ket{\nu^e_i}$, which are now orthogonal thanks to the result of Ref.~\cite{Walgate2000}.
Moreover, as proved in Ref.~\cite{DAriano2014}, product \acp{POVM} in the Jordan-Wigner representation correspond to \ac{LOCC} Fermionic \acp{POVM}.

As the second case, we consider the situation where only one component out of the four $\ket{\psi_E}$, $\ket{\psi_O}$, $\ket{\phi_E}$, $\ket{\phi_O}$ is null.
Perfect discrimination is implementable through \ac{LOCC} in this case as well.
Let us take for instance the vectors $\ket{\psi}=\ket{\psi_E}+\ket{\psi_O}$, and $\ket{\phi} = \ket{\phi_O}$; Alice and Bob firstly measure the parity of their subsystem and if the outcome is even, they know for sure that the system has been prepared in the state $\ket{\psi}$.
Otherwise, the state after the measurement is either $\ket{\psi_O}/\Norm{\psi_O}$ or $\ket{\phi_O}$, and the above strategy for the first case applies.

In the most general case all four components are non-null.
If the two $E$ and $O$ parts are orthogonal---that is when $\braket{\psi_E|\phi_E} = \braket{\psi_O | \phi_O} = 0$---Alice and Bob can measure locally the parity of their systems, thus obtaining the post-measurement states $\ket{\psi'}=\ket{\psi_E}/\Norm{\psi_E}$ and $\ket{\phi'}=\ket{\phi_E}/\Norm{\phi_E}$ if the outcomes are both even, $\ket{\psi'}=\ket{\psi_O}/\Norm{\psi_O}$ and $\ket{\phi'}=\ket{\phi_O}/\Norm{\phi_O}$ if the outcomes are both odd.
Consequently they reduced to the first case.

There is one situation left fulfilling condition~\eqref{eq:orthogen}, i.e., when $\braket{\psi_E|\phi_E}\neq0$ and  $\braket{\psi_O|\phi_O}\neq0$.
This case exhibits the main difference with respect to quantum theory.
Consider for instance the states $1/\sqrt2(\ket{00}_\rA \ket{00}_\rB \pm \ket{01}_\rA \ket{01}_\rB)$.
In this case, the decompositions in Eq.~\eqref{eq:walg} involves bases $\{\ket{\eta_i}_\rB\}$ and $\{\ket{\nu_i}_\rB\}$
where superpositions forbidden by the Fermionic superselection rule appear.
Indeed, one has $i=\pm$ and
\begin{gather*}
    \ket{\pm}_\rA \coloneqq \frac1{\sqrt2}(\ket{00}\pm\ket{01}), \\
    \ket{\eta_\pm}_\rB = \ket{\nu_\mp}_\rB\coloneqq\frac1{\sqrt2}(\ket{00}\pm\ket{01}).
\end{gather*}

The last case can thus not be treated by straightforwardly applying the quantum strategy of Ref.~\cite{Walgate2000}.
The following theorem summarizes what we discussed so far, and shows that it is not possible to perfectly discriminate two states with $\braket{\psi_E|\phi_E}\neq0$ and  $\braket{\psi_O|\phi_O}\neq0$ through \acp{POVM} in
\ac{SEP}, thus neither by means of \ac{LOCC}.

\begin{theorem}\label{th:orthogonal_discriminability}
	Let $\ket{\psi}$ and $\ket{\phi}$ be two pure, normalized and orthogonal states. Then the following statements are equivalent:
	\begin{enumerate*}[label=(\roman*)]
	    \item\label{itm:orthogonal_Fermion} The even and odd parts are separately orthogonal, i.e.,
	\end{enumerate*}	
	\begin{equation}\label{eq:even_discriminable_state}
	    \braket{\psi_E | \phi_E} = \braket{\psi_O | \phi_O} = 0 .
	\end{equation}
	\begin{enumerate*}[label=(\roman*),resume]
	    \item\label{itm:orthogonal_LOCC} The two states are perfectly discriminable through \ac{LOCC}.
		\item\label{itm:orthogonal_SEP} The two states are perfectly discriminable through \ac{SEP}.
	\end{enumerate*}
\end{theorem}
\begin{proof}
	It is trivial to see that \textit{\ref{itm:orthogonal_LOCC}}~$\Rightarrow$~\textit{\ref{itm:orthogonal_SEP}}, whereas we have already shown above that \textit{\ref{itm:orthogonal_Fermion}}~$\Rightarrow$~\textit{\ref{itm:orthogonal_LOCC}} thanks to Ref.~\cite{Walgate2000}.
	We now focus on the implication \textit{\ref{itm:orthogonal_SEP}}~$\Rightarrow$~\textit{\ref{itm:orthogonal_Fermion}} and wonder under what conditions one has
	\begin{equation}\label{eq:operational_norm}
		\max_{S \in \SEPSet} \Tr[(\ket{\psi}\bra{\psi} - \ket{\phi}\bra{\phi}) S]=1;
	\end{equation}
	namely, the condition for perfect discriminability via \ac{SEP}.
	The expression in Eq.~\eqref{eq:operational_norm} clearly involves only the component of $S$ supported on the even subspace $\FockSpace_e(\rA\rB)$.
	Now, a necessary condition for 
	a Fermionic effect $S$ supported on $\FockSpace_e(\rA\rB)$ to be \ac{SEP} is that $S = S_E + S_O$, where $S_E$ and $S_O$ have their support on the $E$ space and $O$ space, respectively (see the Supplemental Material).
	Consequently, the condition in Eq.~\eqref{eq:operational_norm} is equivalent to
	\begin{align}
		\begin{aligned}
		\Tr\Big[\Big(\ket{\tilde\psi_E} \bra{\tilde\psi_E} - \ket{\tilde\phi_E} \bra{\tilde\phi_E}\Big) S_E\Big] &= 1, \\
		\Tr\Big[\Big(\ket{\tilde\psi_O} \bra{\tilde\psi_O} - \ket{\tilde\phi_O} \bra{\tilde\phi_O}\Big) S_O\Big] &= 1 , 
		\end{aligned}
		\label{eq:perfect_condition}
	\end{align}
	for $S=S_E+S_O$ representing an effect in \ac{SEP}, $\ket{\tilde\psi_E}$, $\ket{\tilde\phi_E}$, $\ket{\tilde\psi_O}$, $\ket{\tilde\phi_O}$ being normalized vectors such that $\ket{\tilde\psi_E} = \ket{\psi_E} / \Norm{\psi_E}$ etc.
	Thus, it is possible to perfectly discriminate the two states through separable effects only if the $E$  and $O$ parts are perfectly discriminable separately, as required in Eq.~\eqref{eq:even_discriminable_state}.
\end{proof}

\section{Ancilla assisted discrimination}
We now show that one can overcome the limits of Theorem~\ref{th:orthogonal_discriminability} by providing the two parties with an ancillary system prepared in a suitable pure entangled state $\ket\omega$.
Let us take
\begin{equation}\label{eq:ancilla}
	\ket{\omega}_{\rm AB} \coloneqq a \ket{00} + b \ket{11} \quad\text{for}\quad a, b \neq 0 ,
\end{equation}
and consider the task of discriminating the new vectors $\ket{\psi'} \coloneqq \ket{\psi} \otimes \ket{\omega}$ and $\ket{\phi'} \coloneqq \ket{\phi} \otimes \ket{\omega}$.
In particular, we will see that only a \emph{maximally entangled} ancillary state---i.e., with $\abs a^2=\abs b^2=1/2$---enables perfect discrimination between every two pure Fermionic states, regardless of condition~\eqref{eq:even_discriminable_state}.

\begin{theorem}\label{th:ancilla_assisted}
	It is always possible to perfectly discriminate between every two pure, normalized and orthogonal preparations $\ket{\psi}$ and $\ket{\phi}$ with \ac{LOCC} and an ancillary system in a pure maximally entangled state 
	\begin{equation}
		\ket{\omega}_{\rm AB} = \frac{1}{\sqrt{2}} \left(\ket{00} + e^{i\varphi} \ket{11}\right) , \quad \varphi \in [0, 2\pi) .
		\label{eq:ancimax}
	\end{equation}	
    Moreover, the same does not hold if the ancillary state is not maximally entangled.
\end{theorem}

\begin{proof}
	We show here a sketch of the proof, the full rigorous derivation being given in the Supplemental Material.
	Let us consider the states
	\begin{align*}
		&\ket{\psi'} = \ket{\psi} \otimes \ket{\omega} =\ket{\psi'_O}+\ket{\psi'_E},\\ 
		&\ket{\phi'} = \ket{\psi} \otimes \ket{\omega} =\ket{\phi'_O}+\ket{\phi'_E},
	\end{align*}
	with $\ket{\psi'_E}= 		
		a \ket{\psi_E00} + b \ket{\psi_O 11}$, $\ket{\psi'_O}
		=b \ket{\psi_E 11} + a \ket{\psi_O 00}$, $\ket{\phi'_E}=a \ket{\phi_E 00} + b  \ket{\phi_O 11}$, and $\ket{\phi'_O}=b  \ket{\phi_E 11} + a \ket{\phi_O 00}$, 
	and evaluate for $\abs{a}^2 = \abs{b}^2 = \frac{1}{2}$ the scalar products
	\begin{equation*}
		 \braket{\psi_E' | \phi_E'} =  \braket{\psi_O' | \phi_O'} = \frac{1}{2} \braket{\psi | \phi} = 0 .
	\end{equation*}
	The vectors $\ket{\psi'}$ and $\ket{\phi'}$ do satisfy Eq.~\eqref{eq:even_discriminable_state}, even if $\ket{\psi}$ and $\ket{\phi}$ may not.
	Thus, we are now able to apply the protocol of Theorem~\ref{th:orthogonal_discriminability} to the new states as shown above. Condition~\eqref{eq:ancimax} is also necessary for perfect discrimination, as shown in the Supplemental Material.
\end{proof}

\section{Optimal discrimination}
If the orthogonality condition $\braket{\psi|\phi}=0$ is relaxed,  the two states are clearly not perfectly discriminable.
Hence, one looks for the protocol which minimizes the error probability---i.e., the probability of wrong detection.
For this purpose, it is necessary to introduce our prior probabilities for the two states, given by the distribution $\{p,q\}$. In this case, the error probability reads
\[ \Perr \coloneqq \Tr[p\ket{\psi} \bra{\psi} \Pi_\phi + q\ket{\phi} \bra{\phi} \Pi_\psi] , \]
where $\{\Pi_\psi,\Pi_\phi\}$ is the binary \ac{POVM} describing the discrimination protocol.
We remind that by definition the \ac{POVM} satisfies $\Pi_\psi, \Pi_\phi \ge 0$ and $\Pi_\psi + \Pi_\phi = I$.
In the quantum case, the optimal discrimination strategy corresponds to the \ac{POVM} $\{\ket{+}\bra+, \ket{-}\bra-\}$ diagonalizing the operator
\begin{equation}\label{eq:operator_delta}
	\Delta \coloneqq p \ket{\psi} \bra{\psi} - q \ket{\phi} \bra{\phi} = \lambda_+ \ket{+} \bra{+} + \lambda_- \ket{-} \bra{-} ,
\end{equation}
where $\lambda_+ > 0$, $\lambda_- < 0$ are the eigenvalues of $\Delta$, and $\braket{+ | -} = 0$ (see Refs.~\cite{Helstrom1969,Helstrom1976}).
The corresponding error probability is~\cite{Helstrom1976}
\begin{align}
\Perr=\frac12\left(1-{\Norm\Delta_1}\right).
\label{eq:hell}
\end{align}
In Ref.~\cite{Virmani2001}, the authors observe that optimal discrimination through \ac{LOCC} of $\ket\psi$ and $\ket\phi$ with prior probabilities $p$ and $q$, respectively, is equivalent to perfect \ac{LOCC} discrimination between $\ket{+}$ and $\ket{-}$ (see also Ref.~\cite{Helstrom1976}), thus reducing the optimal case to an instance of perfect discrimination.
While the latter is always possible in quantum theory, we know from Theorem~\ref{th:orthogonal_discriminability} that in Fermionic theory this is true only if the eigenvectors satisfy
\begin{equation}\label{eq:optimal_LOCC_condition}
	\braket{+_E | -_E} = \braket{+_O | -_O} = 0.
\end{equation}
Otherwise, by Theorem~\ref{th:ancilla_assisted} perfect \ac{LOCC} discrimination requires a maximally entangled ancilla.
As for the perfect discrimination case, also the conditions for optimal \ac{LOCC} discrimination in Fermionic theory differ from the quantum ones only when the $E$ and $O$ components of $\ket+$ and $\ket-$ are all non-zero, and $\braket{+_E|-_E},\braket{+_O|-_O}\neq0$.
For the latter case, we now prove a necessary and sufficient condition for achievability of optimal discrimination with \ac{LOCC} that does not require diagonalization of $\Delta$.

\begin{theorem}\label{th:Optimal_LOCC}
	Let $\rho = p \ket{\psi} \bra{\psi}$ and $\sigma = q \ket{\phi} \bra{\phi}$ be two pure and sub-normalized states for $p, q > 0$ and $p + q = 1$.
	They are optimally discriminable through \ac{LOCC} if and only if they satisfy
	\begin{equation}\label{eq:optimal_condition}
		[\Delta, P_E] = 0 ,
	\end{equation}
	where $\Delta$ is defined in Eq.~\eqref{eq:operator_delta} and $P_E$ is the projector onto the $E$ subspace.
\end{theorem}

\begin{remark}
    Let us consider the projectors $P_e$ and $P_O$ on the even subspace $\FockSpace_e(\rA\rB)$ and $O$ subspace of system $\rA\rB$, respectively, to observe that $P_e = P_E + P_O$ and $[\Delta, P_e] = 0$.
    Hence, Eq.~\eqref{eq:optimal_condition} is fulfilled if and only if $[\Delta, P_O] = 0$ so the two expressions are interchangeable.
\end{remark}

\begin{proof}
	Since optimal discrimination between $\ket\psi$ and $\ket\phi$ is equivalent to perfect discrimination between $\ket+$ and $\ket-$, by Eq.~\eqref{eq:optimal_LOCC_condition} optimal discriminability of the states $\ket\psi$ and $\ket\phi$ by \ac{LOCC} is equivalent to the condition
	\begin{equation}\label{eq:Fermionic_optimal_condition_v2}
		\braket{+ | P_E | -} = \braket{+ | P_O | -} = 0 .
	\end{equation}
	Now, taking the difference of the first two members of Eq.~\eqref{eq:Fermionic_optimal_condition_v2}, we can then express the \ac{LOCC}-discriminability condition through the single expression
		\begin{equation}\label{eq:Fermionic_optimal_condition_v3}
		\braket{+ |( P_E - P_O )| -} = 0 .
	\end{equation}
	Indeed, since $P_O = P_e - P_E$, Eq.~\eqref{eq:Fermionic_optimal_condition_v3} is equivalent to the requirement that the restriction of the projector $P_E$ onto the space $\Span\{\ket{\psi}, \ket{\phi}\}$ is diagonal in the basis $\{\ket{+}, \ket{-}\}$.
	The operators $\Delta$ and $P_E$ are simultaneously diagonalizable if and only if $[\Delta, P_E] = 0$.
	Equation \eqref{eq:optimal_condition} is then equivalent to attainability of optimal discrimination between the two states $\rho$ and $\sigma$ via \ac{LOCC}.
\end{proof}

We may wonder what happens when condition~\eqref{eq:optimal_condition} is not satisfied.
As we show in the next theorem, the best discrimination strategy through \ac{SEP} corresponds to measuring in the basis of eigenvectors of $\Delta_E$ and $\Delta_O$, defined as the restriction of the operator $\Delta$ onto the $E$ and $O$ subspaces, respectively.
Such a strategy is \ac{LOCC}.

\begin{theorem}\label{th:optimal_LOCC_error}
	Let $\rho = p\ket\psi \bra\psi$ and $\sigma = q \ket\phi \bra\phi$ be two pure subnormalized states for $p, q > 0$ and $p + q = 1$.
	The optimal \ac{SEP} discrimination protocol is locally implementable through \ac{LOCC} and its error probability reads
	\begin{equation}\label{eq:optimal_LOCC_error}
		\Perr^\SEPSet = \Perr^\LOCCSet = \frac{1}{2} (1 - \Norm{\Delta_E + \Delta_O}_1) ,
	\end{equation}
	where $\Delta_E = P_E \Delta P_E$ and $\Delta_O = P_O \Delta P_O$.
\end{theorem}

\begin{proof}
    This result can be obtained considering that
    \begin{equation*}
    	\Perr^\SEPSet = p - \max_{\Pi_\psi\in\SEPSet(\rA\rB)} \Tr[\Pi_\psi \Delta] ,
    \end{equation*}
    where $\Pi_\psi$ must be of the form $\Pi_\psi = \Pi_\psi^E + \Pi_\psi^O$ in order to comply with the separability condition, as observed in the proof of Theorem~\ref{th:orthogonal_discriminability}.
    The result then follows.
\end{proof}

The above result allows us to treat the case where we are restricted only to local measurements and Eq.~\eqref{eq:optimal_condition} does not hold for the preparations $\rho$, $\sigma$.
Once we are given the pure states $\ket\psi$ and $\ket\phi$, the condition for optimal \ac{LOCC} discrimination of Eq.~\eqref{eq:optimal_condition} is fulfilled either for the vectors laying in the $E$ or $O$ space, i.e., $[\ket\psi\bra\psi, P_E] = [\ket\phi\bra\phi, P_E] = 0$, or if the probability $p$ satisfies
\begin{equation}\label{eq:unstable}
	[\ket\psi\bra\psi, P_E] = \frac{1 - p}{p} [\ket\phi\bra\phi, P_E] .
\end{equation}
Condition~\eqref{eq:unstable} can be satisfied by a unique value of the prior probability $p$, unless $[\ket\psi\bra\psi,P_E]=[\ket\phi\bra\phi,P_E]=0$.
However, we now show that optimal \ac{LOCC} discrimination can achieve the performances of unconstrained protocols, provided that two ancillary Fermionic systems are used in a maximally entangled state.
As discussed above, indeed, the problem of optimal discrimination between two pure states reduces to that of the orthogonal vectors $\ket+,\ket-$ in Eq.~\eqref{eq:operator_delta}.
Considering Theorem~\ref{th:ancilla_assisted}, we know that orthogonal states can be perfectly discriminated via \ac{LOCC} provided a maximally entangled ancillary system is available. These two observations immediately lead to our last result.
\begin{theorem}
	Let $\rho = p\ket\psi \bra\psi$ and $\sigma = q \ket\phi \bra\phi$ be two pure subnormalized states for $p, q > 0$ and $p + q = 1$.
	It is always possible to optimally discriminate between the two preparations via \ac{LOCC} if we use an ancillary system in a pure maximally entangled state.
\end{theorem}

Equation~\eqref{eq:unstable} introduces a strict condition on the prior probability of the preparations, which are always subject to noise.
We show hereafter that if we introduce a small perturbation $\epsilon$ on the preparation probabilities of pair of states satisfying Eq.~\eqref{eq:optimal_condition}, the discrimination error probability increases at most linearly in $\epsilon$ with respect to the appropriate optimal \ac{LOCC} protocol.
Thus, we map $p \mapsto p + \epsilon$ and attain
\begin{equation*}\begin{split}
	\Delta^\epsilon \coloneqq& (p + \epsilon) \ket{\psi} \bra{\psi} - (q - \epsilon) \ket{\phi} \bra{\phi} \\
	=& \Delta^{0} + \epsilon (\ket{\psi} \bra{\psi} + \ket{\phi} \bra{\phi}) ,
\end{split}\end{equation*}
where $[\Delta^0, P_E] = 0$.
At this stage, we estimate the error difference between the optimal \textsc{povm} $\mathbb{P}^0 \coloneqq \{\Pi_\psi, \Pi_\phi\}$ for $\epsilon = 0$, which is \ac{LOCC} thanks to Theorem~\ref{th:Optimal_LOCC}, and the \ac{LOCC}-optimal \textsc{povm} for the perturbed case $\Delta^\epsilon$.
The error increases as $\delta\Perr \coloneqq \Perr(\mathbb{P}^0 | \Delta^\epsilon) - \Perr^\LOCCSet(\Delta^\epsilon) \ge 0$ where $\Perr(\mathbb{P}^0 | \Delta^\epsilon) = \Tr[(p + \epsilon) \ket\psi\bra\psi \Pi_\phi + (q - \epsilon) \ket\phi\bra\phi \Pi_\psi]$ and $\Perr^\LOCCSet(\Delta^\epsilon) = \frac{1}{2} (1 - \Norm{\Delta_E^\epsilon + \Delta_O^\epsilon}_1)$ as in Eq.~\eqref{eq:optimal_LOCC_error}.
Accordingly manipulating the expression for $\delta\mathcal{P}_{\text{err}}$ one obtains
\begin{align}\label{eq:bound}
	\delta\Perr \leq k\abs\epsilon+g\epsilon,
\end{align}
where $k,g\geq0$ are suitable constants depending only on $\ket\psi$, $\ket\phi$.
The former inequality is as tight as possible: let us take indeed the states $\ket\psi = 1/\sqrt{2} \ket{00} + 1/\sqrt{2} \ket{11}$ and $\ket\phi = \alpha \ket{00} + \sqrt{1-\alpha^2}\ket{11}$, where $\alpha \coloneqq (1/\sqrt{2} + \xi)$, and $\xi$ belongs to a neighborhood of zero.
In such a case, we have numerically assessed that the error difference $\delta\Perr$ exhibits a corner in $\epsilon = 0$ as $\xi \to 0$ (more details can be found in the Supplemental Material).

We also investigate the performance of the optimal \ac{LOCC} protocol for $\epsilon\neq0$ in the neighborhood of a prior probability $p$ satisfying condition~\eqref{eq:unstable}, by comparing its efficiency to that of the optimal unconstrained (i.e., entanglement-assisted \ac{LOCC}) \textsc{povm}.
Thus we estimate $\delta\Perr' \coloneqq \Perr^\LOCCSet(\Delta^\epsilon) - \Perr(\Delta^\epsilon) \ge 0$ by means of Eqs.~\eqref{eq:hell} and~\eqref{eq:optimal_LOCC_error}, obtaining
\begin{equation}\begin{split}
	\delta\Perr' &\leq \kappa\abs\epsilon ,
\end{split}\label{eq:boundlocc}
\end{equation}
for a suitable $\kappa\geq0$, thanks to the triangle inequality.

We remark that, in the case of a mismatch in the assessment of the prior probability $p$,
also for unconstrained optimal discrimination---coinciding with ancilla-assisted \ac{LOCC}---one has the same bound as in Eq.~\eqref{eq:bound}, with possibly different constants $k$ and $g$.
This feature, however, must not be considered as an artefact of Fermionic theory.
Indeed, the technique used to derive the bound in Eq.~\eqref{eq:bound} is very general and leads to the same behavior in the quantum case as well.

\section{Discussion}
As in the quantum case, discrimination with separable and \ac{LOCC} \acp{POVM} in the Fermionic case achieve the same performances.
Unlike in quantum theory, on the other hand, in Fermionic theory ideal state discrimination through \ac{LOCC} is subject to non-trivial conditions.
In this Letter, we derived the conditions under which \ac{LOCC} discrimination achieves the ideal performances of unconstrained discrimination protocols.
However, in the Fermionic case, ancilla-assisted \ac{LOCC} protocols achieve ideal discrimination.
One has to remark, though, that this is the case only for maximally entangled ancillary states.
The former statement unequivocally determines the amount of entanglement required for such a task. 
We finally studied the behavior of optimal protocols---which depend on prior probabilities of the states to be discriminated---if the prior conditions are subject to perturbation.
A remarkable instability is observed, corresponding to a corner point in the curve representing the error probability excess due to non-optimized \acp{POVM}.
We stress that the latter phenomenon is not exclusive of Fermionic theory, as it occurs also in the quantum case.

\begin{acknowledgments}
	We thank Massimiliano F.~Sacchi for useful discussions and comments.
	A.~T. acknowledges financial support from Elvia and Federico Faggin foundation through Silicon Valley Community Foundation, Grant No. 2020-214365.
\end{acknowledgments}

\bibliography{addendum}

\bigskip
\appendix

\section{Separable effects}
In order to implement the parity superselection rule, the operator $0\leq S\leq I$ representing a separable effect supported on $\FockSpace_e(\rA\rB)$ must be of the form
\begin{equation}
	S = S_E + S_O, 
\end{equation}
where $S_E= \sum_i e_i \otimes e_i'$, $S_O=\sum_j o_j \otimes o_j'$, and $e_i, e_i', o_j, o_j' \ge 0$, with
\begin{align*}
	\Supp(e_i) &\subseteq \FockSpace_e (\rA) & \Supp(e_i') &\subseteq \FockSpace_e (\rB) \\
	\Supp(o_j) &\subseteq \FockSpace_o (\rA) & \Supp(o_j') &\subseteq \FockSpace_o (\rB) .
\end{align*}
Once the effect is applied to a Fermionic state $\tau \in \Stset(\rA\rB)$, the Born rule returns
\begin{equation}
	\Tr[\tau S] = \Tr[P_E \tau P_E S_E + P_O \tau P_O S_O] .
\end{equation}
The above expression shows that any separable \ac{POVM} operates on the $E$ and $O$ parts of $\tau$ independently.
In particular, in the proof of Theorem~\ref{th:orthogonal_discriminability} we seek the maximum of $r \coloneqq \Tr[(\rho - \sigma) S]$, with $\rho = \ket\psi\bra\psi$ and $\sigma = \ket\phi
\bra\phi$, i.e.,
\begin{align*}
	r =& \Tr\Big[\left( \ket{\psi_E} \bra{\psi_E} -  \ket{\phi_E} \bra{\phi_E}\right) S_E \\
	&+ \left( \ket{\psi_O} \bra{\psi_O} - \ket{\phi_O} \bra{\phi_O}\right) S_O\Big].
\end{align*}
The latter achieves unit value if and only if one can find $S_E$ and $S_O$ such that
\begin{align*}
    \Tr[\rho S] &= \Norm{\psi_E}^2 \braket{\tilde\psi_E | S_E | \tilde\psi_E} + \Norm{\psi_O}^2 \braket{\tilde\psi_O | S_O |\tilde\psi_O} = 1, \\
    \Tr[\sigma S] &= \Norm{\phi_E}^2 \braket{\tilde\phi_E | S_E | \tilde\phi_E} + \Norm{\phi_O}^2 \braket{\tilde\phi_O | S_O | \tilde\phi_O} = 0 , 
\end{align*}
where $\ket{\tilde\eta}\coloneqq\ket\eta/\Norm\eta$.
However, due to the hypotheses assumed so far, we achieve the above conditions if and only if Eq.~\eqref{eq:perfect_condition} is satisfied.

\section{Proof of Theorem~\ref{th:ancilla_assisted}}

Alice and Bob are provided with an entangled ancilla in the state $\ket{\omega}$, as in Eq.~\eqref{eq:ancilla}.
They now share two bipartite systems in the possible states $\ket{\psi'} = \ket{\psi} \otimes \ket{\omega}$ or $\ket{\phi'} = \ket{\phi} \otimes \ket{\omega}$, whose full expression can be obtained from 
\begin{equation}\label{eq:ancilla_pxi_prime}
	\begin{aligned}
		\ket{\psi'_E} &= a \sum_{i = 0}^n \ket{e_i 0}_\rA \ket{\eta_i^e 0}_\rB + b  \sum_{j = 0}^n \ket{o_j 1}_\rA \ket{\eta_j^o 1}_\rB \\
		\ket{\phi'_E} &= a \sum_{i = 0}^n \ket{e_i 0}_\rA \ket{\nu_i^e 0}_\rB + b  \sum_{j = 0}^n \ket{o_j 1}_\rA \ket{\nu_j^o 1}_\rB \\
        \ket{\psi'_O} &= b \sum_{i = 0}^n \ket{e_i 1}_\rA \ket{\eta_i^e 1}_\rB + a  \sum_{j = 0}^n \ket{o_j 0}_\rA \ket{\eta_j^o 0}_\rB \\
		\ket{\phi_O'} &= b \sum_{i = 0}^n \ket{e_i 1}_\rA \ket{\nu_i^e 1}_\rB + a\sum_{j = 0}^n \ket{o_j 0}_\rA \ket{\nu_j^o 0}_\rB .
    \end{aligned}
\end{equation}
Let $\Sigma_E\coloneqq \braket{\psi_E|\phi_E}$ and $\Sigma_O\coloneqq \braket{\psi_O|\phi_O}=-\Sigma_E$, where the last equality follows from the fact that  $\Sigma_E+\Sigma_O=\braket{\psi|\phi}=0$.
As shown in the body, there are cases where the ancilla is not needed, and clearly its presence cannot reduce the performances of \ac{LOCC} discrimination. The remaining case is that where  $\Sigma_E\neq0$.
The necessary and sufficient condition for perfect \ac{LOCC} discrimination between $\ket{\psi'}$ and $\ket{\phi'}$ of Eq.~\eqref{eq:even_discriminable_state} can then be written using Eq.~\eqref{eq:ancilla_pxi_prime} as 
\begin{equation*}
    \braket{\psi_E'|\phi_E'} = (\abs{a}^2 - \abs{b}^2)\Sigma_E=0.
\end{equation*}
For $\abs{a}^2 = \abs{b}^2 = \frac{1}{2}$ the above condition is clearly satisfied.
On the other hand, if $\Sigma_E\neq0$, discrimination by \ac{LOCC} is not possible for $\abs{a}\neq\abs{b}$.

\section{Extremal case for $\delta\Perr$ bound}
In the Letter we investigated the behavior of the discrimination error in the case where the prior probabilities slightly differ from the ideal ones.
We are given two pure states $\ket\psi$, $\ket\phi$ and if there exists a probability distribution $\{p, q\}$ such that condition~\eqref{eq:unstable} is satisfied, we proved that such a solution is unique and the optimal discrimination strategy is \ac{LOCC}-implementable, unless $[P_E,\ket\psi\bra\psi]=[P_E,\ket\phi\bra\phi]=0$.
Therefore, a small perturbation $\epsilon$ in the prior probability $p$ produces an increase of error probability of the \text{locc} protocol---which is optimized for the unperturbed case---with respect to the optimal \ac{LOCC} one.
For this purpose, we introduce the quantity
\begin{equation}\label{eq:Perr}
	\delta\Perr \coloneqq \Perr(\mathbb{P}^0 | \Delta^\epsilon) - \Perr^\LOCCSet(\Delta^\epsilon).
\end{equation}
Thanks to the triangle inequality we have that 
\begin{align*}
    \delta\Perr =&\frac12\left(\abs{\Norm{\Delta_E^\epsilon + \Delta_O^\epsilon}_1 - \Norm{\Delta^0}_1}\right. \\
    &\left.+\epsilon\Tr[(\ket\psi\bra\psi+\ket\phi\bra\phi)(\Pi_\phi-\Pi_\psi)]\right) \\
    &\le k\abs\epsilon+g\epsilon, \\
\intertext{for}
    k \coloneqq& \frac12\norm{\delta^\epsilon_E+\Delta^\epsilon_O-\Delta^0},\\
    g \coloneqq& \frac12\Tr[(\ket\psi\bra\psi+\ket\phi\bra\phi)(\Pi_\phi-\Pi_\psi)].
\end{align*}
Hence, the error difference $\delta\mathcal{P}_{\text{err}}$ is \textsl{sublinear}.

We numerically assessed that the bound above is indeed achieved by the states
\begin{align*}
	\ket\psi &= \frac{1}{\sqrt{2}} \ket{00} + \frac{1}{\sqrt{2}} \ket{11} \\
	\ket\phi &= \left(\frac{1}{\sqrt{2}} + \xi\right) \ket{00} + \frac{\gamma}{\sqrt2} \ket{11},
\end{align*}
where $\gamma \coloneqq \sqrt{1 - 2\sqrt{2} \xi - 2 \xi^2}$ and $\xi$ belongs to a neighborhood of zero.
The condition for optimality of Eq.~\eqref{eq:unstable} is fulfilled by
\begin{equation*}
	p(\xi) = \frac{\gamma + \sqrt2 \gamma \xi}{1 + \gamma + \sqrt2 \gamma \xi} \quad \text{for} \quad \xi \in [0, 1 - \sqrt2/2)
\end{equation*}
and the terms of Eq.~\eqref{eq:Perr} read
\begin{align*}
	\Perr(\mathbb{P}^0 | \Delta^\epsilon) &= \Tr[(p + \epsilon) \ket\psi\bra\psi \Pi_\phi + (q - \epsilon) \ket\phi\bra\phi \Pi_\psi] \\
	\Perr^\LOCCSet(\Delta^\epsilon) &= \frac{1}{2} (1 - \Norm{\Delta_E^\epsilon + \Delta_O^\epsilon}_1) .
\end{align*}
In Fig.~\ref{fig:delta_error} we show a plot of the quantity $\delta\Perr$ versus $\epsilon$ and $\xi$.
We observe that, letting $\xi$ vary in a neighborhood of 0 one gets arbitrarily close to the bound in Eq.~\eqref{eq:bound}.
On the other hand, the same analysis shows that one cannot find any lower bound for $\delta\Perr$ better than $\delta\Perr\geq0$.
Following exactly the same line as in the above derivation of the bound in Eq.~\eqref{eq:bound}, one can derive the bound in Eq.~\eqref{eq:boundlocc}.

\begin{figure}
	\centering
	\includegraphics[width=8.6cm]{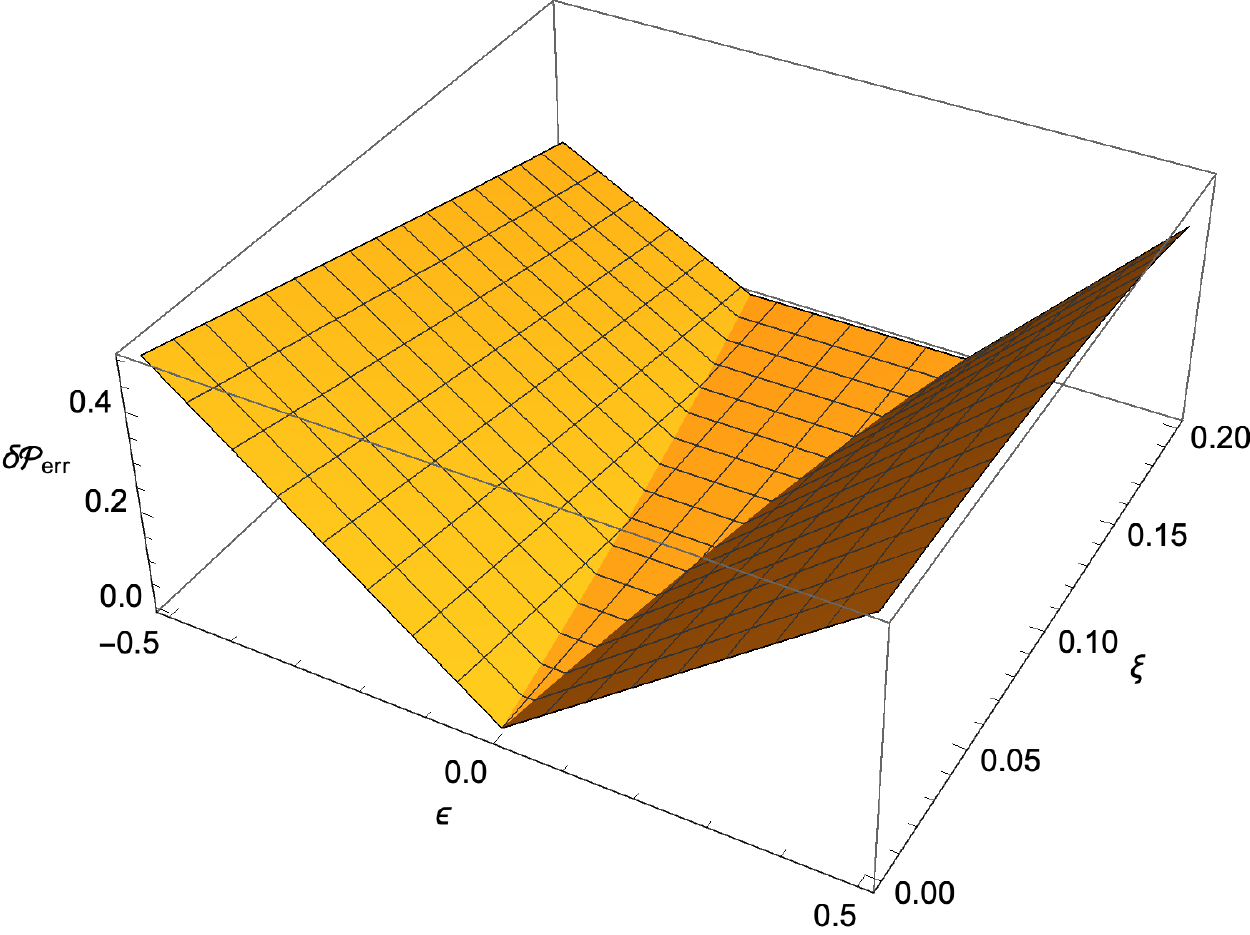}
	\caption{The plot shows the difference $\delta\Perr$ between the error probability $\Perr(\mathbb{P}^0 | \Delta^\epsilon)$ in discrimination between $\rho = (p+\epsilon)\ket\psi\bra\psi$ and $\sigma = (1-p-\epsilon)\ket\phi\bra\phi$ with the \ac{POVM} that is optimal for discrimination between $p\ket\psi\bra\psi$ and $(1-p)\ket\phi\bra\phi$ and the error probability $\Perr^\LOCCSet(\Delta^\epsilon)$ in discrimination between the same states $\rho$ and $\sigma$ with the correct \ac{LOCC}-optimal \ac{POVM}, as a function of $\epsilon$ and $\xi$, where $\ket\psi=1/{\sqrt{2}} (\ket{00} +\ket{11})$ and $\ket\phi = \alpha \ket{00} + \sqrt{1-\alpha^2}\ket{11}$, and $\alpha=1/{\sqrt{2}} + \xi$.
	The special value $p$ of the prior probability, corresponding to $\epsilon=0$, is such that $p\ket\psi\bra\psi$ and $(1-p)\ket\phi\bra\phi$ are ideally discriminable via \ac{LOCC}.
	All the other values of $\epsilon$, on the other hand, lead to pairs of states $\rho$ and $\sigma$ that cannot be ideally discriminated via \ac{LOCC}.
	For values of $\xi$ in a neighborhood of 0 the function $\delta\Perr$ gets arbitrarily close to the bound $\delta\Perr\leq k\abs\epsilon+g\epsilon$ for suitable constants $k$ and $g$.}
	\label{fig:delta_error}
\end{figure}
\end{document}